\documentclass[onefignum,onetabnum]{siamart171218}

\usepackage{lipsum}
\usepackage{amsmath}
\usepackage{amsfonts}
\usepackage{mathrsfs}
\usepackage{bm}
\usepackage{graphicx}
\usepackage{epstopdf}
\usepackage{algpseudocode}
\usepackage{hyperref}
\ifpdf
  \DeclareGraphicsExtensions{.eps,.pdf,.png,.jpg}
\else
  \DeclareGraphicsExtensions{.eps}
\fi

\newsiamremark{remark}{Remark}
\newsiamremark{hypothesis}{Hypothesis}
\crefname{hypothesis}{Hypothesis}{Hypotheses}
\newsiamthm{claim}{Claim}

\headers{Fast tunable blurring algorithm for scattered data}{G. Robinson \& I. Grooms}

\title{A fast tunable blurring algorithm \\
for scattered data \thanks{Submitted to the editors June 16, 2019.
\funding{This work was funded by the National Science Foundation under award no.~1821074.}}}

\author{Gregor Robinson \&
  Ian Grooms
  \thanks{Department of Applied Mathematics, University of Colorado Boulder, Boulder, CO
  (\email{gregor.robinson@colorado.edu}, \email{ian.grooms@colorado.edu}).}}

\usepackage{amsopn}

\DeclareMathOperator{\vspan}{span}

\makeatletter
\newcommand*{\addFileDependency}[1]{
  \typeout{(#1)}
  \@addtofilelist{#1}
  \IfFileExists{#1}{}{\typeout{No file #1.}}
}
\makeatother

\algblock{Input}{EndInput}
\algnotext{EndInput}
\algblock{Output}{EndOutput}
\algnotext{EndOutput}
\newcommand{\Desc}[2]{\State \makebox[2em][l]{#1}#2}

\ifpdf
\hypersetup{
  pdftitle={A fast tunable blurring algorithm for scattered data},
  pdfauthor={G. Robinson and I. Grooms}
}
\fi

\begin{document}

\maketitle

\begin{abstract}
A blurring algorithm with linear time complexity can reduce the small-scale content of data observed at scattered locations in a spatially extended domain of arbitrary dimension.
The method works by forming a Gaussian interpolant of the input data, and then convolving the interpolant with a multiresolution Gaussian approximation of the Green's function to a differential operator whose spectrum can be tuned for problem-specific considerations.
Like conventional blurring algorithms, which the new algorithm generalizes to data measured at locations other than a uniform grid, applications include deblurring and separation of spatial scales.
An example illustrates a possible application toward enabling importance sampling approaches to data assimilation of geophysical observations, which are often scattered over a spatial domain, since blurring observations can make particle filters more effective at state estimation of large scales.
Another example, motivated by data analysis of dynamics like ocean eddies that have strong separation of spatial scales, uses the algorithm to decompose scattered oceanographic float measurements into large-scale and small-scale components.
\end{abstract}



\section{Introduction}
\label{sec:intro}
We present a linear-time blurring algorithm that works for data measured at scattered points in $\mathbb{R}^d$, for arbitrary dimension $d>0$, and that permits choice in shaping its response to different length scales.
This algorithm is equivalent to solving a positive definite self-adjoint elliptic partial differential equation.

By blurring, we mean the process of attentuating small-scale content of a dataset.
Doing so offers a way to denoise and simplify spatial data whose large features are of primary relevance.
Blurring is also a mechanism for isolating small-scale content by simply subtracting the blurred version from the original.
We show an example of each of these use cases.

Applications of blurring in scientific computing generally benefit from having control over the blurring spectrum, i.e.~the factor by which it attenuates inputs of different spatial scales.
This concept is made more rigorous in \ref{sec:method}.
Regularly-spaced data on a periodic domain would enable straightforward application of Fourier methods to implement a blurring algorithm that obeys a desired spectrum.
But applications in geophysics, and remote sensing in general, often involve measurements made at irregularly scattered locations in a spatially-extended domain.
It is therefore desirable not to require a regular grid.

The paper is organized as follows.
Our blurring method is described in \cref{sec:method}; an illustrative synthetic example is shown in \cref{sec:ex1}; an example of separating large and small scales of scattered oceanographic data is presented in \cref{sec:scale-separation}; how this blur can be applied in its original motivating context of meteorological data assimilation is described in \cref{sec:ssir} along with a connection to generalized Gaussian random fields that led to our algorithm's discovery; another example using real meteorological data is in \cref{sec:ex-radiosonde} to show the blur has the desired effect on particle filtering; algorithmic complexity and generalizations toward practical application of our method are discussed in \cref{sec:discussion}; and conclusions follow in \cref{sec:conclusions}.

\section{Method}
\label{sec:method}
Let $\mathbf{z} \in \mathbb{R}^{N_z}$ be a vector of data at locations $\mathbf{q}_i \in \mathbb{R}^d$ for each $i \in (1, \cdots, N_z)$, where $N_z$ is the number of observations.
The proposed blur $\mathbf{S}$ works by solving for a discrete approximation of $\mathscr{D}^{-1}\zeta$, where $\mathscr{D}$ is an elliptic differential operator described in the next paragraph and $\zeta : \mathbb{R}^{d} \rightarrow \mathbb{R}$ is a continuous-domain interpolant of the data $\mathbf{z}$ expressed as a sum of Gaussians.
This obtains by approximating the Green's function $g$ of $\mathscr{D}$ as a multiresolution sum of Gaussians, computing the convolution of that approximation with $\zeta$, and evaluating the result at the locations $\{\mathbf{q}_i\}$.

Define $\mathscr{D}$ to be the fractional bound-state Helmholtz operator:
\begin{align}
	\mathscr{D} &= \left(1- \ell^2 \Delta \right)^{\beta}, \label{eq:diff_op}
\end{align}
where $\Delta$ is the formal Laplacian operator, $\ell>0$ is a tuning parameter with dimensions of length, and $\beta>0$ is a dimensionless tuning parameter that controls the rate of growth of eigenvalues.
Eigenfunctions of $\mathscr{D}$ are Fourier modes of wavenumber $k$ and corresponding eigenvalues $(1+\ell^2|k|^2)^{\beta}$.
The characteristic scale of this operator is $\ell/(2\pi \sqrt{2^{1/\beta}-1})$, in the sense that eigenfunctions with length scales longer than this have corresponding eigenvectors close to 1.

A drawback of this approach is that convolution with $g$ attenuates all but constant functions on $\mathbb{R}^d$, so even a constant data vector $\bm{z}$ will be attenuated to some degree.
We will discuss a way to mitigate this effect in \cref{sec:ex1}.

To represent the data in a continuous form that allows convolution with $g$, we choose radial basis function (RBF) interpolation \cite{FF15}.
RBF interpolation of the observations requires us to choose a kernel $\psi : \mathbb{R}^+ \rightarrow \mathbb{R}$ that is used as the \emph{radial basis} in which the interpolant will represent the data.
The interpolant takes the form
\begin{align}
    \zeta(\cdot) &= \sum_{j=1}^{N_z} b_j \psi(\|\cdot - \mathbf{q}_j\|),
\end{align}
where $(b_j)$ are interpolation weights such that
\begin{align}
    \zeta(\mathbf{q}_i) &= \sum_{j=1}^{N_z} b_j \psi(\|\mathbf{q}_i - \mathbf{q}_j\|) = z_i. \label{eq:rbf}
\end{align}
In matrix form, this linear system becomes
\begin{equation}
    \mathbf{B}\mathbf{b} = \mathbf{z}.
\end{equation}
We take the RBF kernel to be $\psi(\|\cdot\|) = \phi(\cdot; 0,\xi \mathbf{I})$, where $\phi(\cdot;\mathbf{\mu},\mathbf{\Sigma})$ is the density of a $d$-variate Gaussian random variable with mean $\mathbf{\mu}$ and covariance $\mathbf{\Sigma}$
\begin{align}
    \phi(\cdot; \bm{\mu}, \bm{\Sigma}) &=  \left(2\pi \det \bm{\Sigma} \right)^{-d/2} \exp \left( (\cdot - \bm{\mu})^T \bm{\Sigma}^{-1} (\cdot - \bm{\mu}) \right).
\end{align}
This notation is used as a convenient description of Gaussian functions even though we will not use them to describe any random variables.

One may worry that this method is hardly fast, despite using a fast PDE solver, since a naive approach to solving for $\mathbf{b}$ requires $\mathcal{O}(N_y^3)$ operations.
Computational complexity of our algorithm, including faster alternatives to solving for $\bm{b}$, is discussed in \cref{sec:discussion}.

The multiresolution Gaussian approximation of the Green's function begins by writing the Fourier transform of the inverse of \cref{eq:diff_op} as an inverse-power function $\hat{g}(t)=t^{-\beta}$ where $t = 1 + \ell^2 |k|^2$.
Exponential approximations of inverse power functions like this are studied in \cite{BM10,McLean18}.
The approach therein is to write a finite trapezoid-type discretization of an integral representation of $\hat{g}$.
With the change of variables introduced by McLean \cite{McLean18}, the integral representation to be discretized is
\begin{align}
    \frac{1}{t^\beta} &= \frac{1}{\Gamma(\beta)} \int_{-\infty}^{\infty} \exp \left(- \varphi ( x,t ) \right) (1 + e^{-x})~dx, \label{eq:mclean-int}
\end{align}
where
\begin{align}
    \varphi( x, t) &= t \exp (x - e^{-x}) - \beta(x-e^{-x}).
\end{align}

The finite trapezoid rule discretizes this into
\begin{align}
    \frac{1}{t^{\beta}} \approx
     \frac{1}{\Gamma(\beta)} \sum_{n=-M_-}^{M_+} v_n e^{-a_n t}, \label{eq:mclean-disc}
\end{align}
where $\Gamma(\beta)$ is the Gamma function and
\begin{align} \label{eq:exp-approx}
    a_n &= \exp(nh - e^{-nh}), \\
    v_n &= h(1+e^{-nh}) \exp \left( \beta(nh - e^{-nh}) \right).
\end{align}


Ref.~\cite{McLean18} Lemma 4 shows that the total required number of terms $M_-+M_++1$ scales as $(\ln E)^2$ to achieve uniform relative error bounded by $E>0$ in the limit $E \downarrow 0$.

The approximation in \cref{eq:mclean-disc} can now be rewritten in terms of normalized multivariate isotropic Gaussian functions of $k$.
Given weights $v_n$ and exponential rates $a_n$ from the exponential approximation above, we can derive the multiplicative factors required of this equivalent formulation:
\begin{align}
    v_n e^{-a_n (1+\ell^2 k^2)} &= v_n e^{-a_n} e^{-a_n \ell^2 k^2} \\
            = v_n e^{-a_n} &(2\pi /2\ell^2a_n)^{d/2} \left( (2\pi /2\ell^2a_n)^{-d/2} e^{-2\ell a_n^2k^2/2} \right) \\
            = v_n e^{-a_n} &(\pi/\ell^2a_n)^{d/2} \phi(k;0,1/2\ell^2a_n). \label{eq:gaussian-mra}
\end{align}
The second line obtains from simultaneously multiplying and dividing by the constant required to normalize the Gaussian term in large parentheses, which is written in that manner to ease visual comparison to the standard form of an isotropic $d$-variate Gaussian probability density function of mean $0$ and variance $1/2\ell^2a_n$.
Combining \cref{eq:mclean-disc}-\cref{eq:gaussian-mra} yields
\begin{align}
    \frac{1}{(1+\ell^2 |k|^2)^{\beta}} \approx
        \frac{1}{\Gamma(\beta)} \sum_{n=-M_-}^{M_+} v_n e^{-a_n} (\pi/\ell^2 a_n)^{d/2} \phi(k; 0, 1/2\ell^2a_n). \label{eq:mclean-disc-gauss}
\end{align}
A plot of the relative error committed by this approximation is shown in \cref{fig:gaussian-approx-error}.

Taking the inverse Fourier transform and combining terms finally yields the desired approximation of the Green's function in physical space in terms of normalized Gaussians:
\begin{align}
    g(\cdot) &\approx \sum_{n=-M_-}^{M_+} c_n \phi(\cdot; 0, \rho_n), \label{eq:green-approx} \\
    \rho_n &= 2 \ell^2 a_n, \\
    c_n &= \frac{v_n (\pi/\ell^2a_n)^{d/2-1}}{\Gamma(\beta) e^{a_n}}. \label{eq:green-approx-end}
\end{align}

\begin{figure}
    \centering
    \includegraphics[width=\columnwidth]{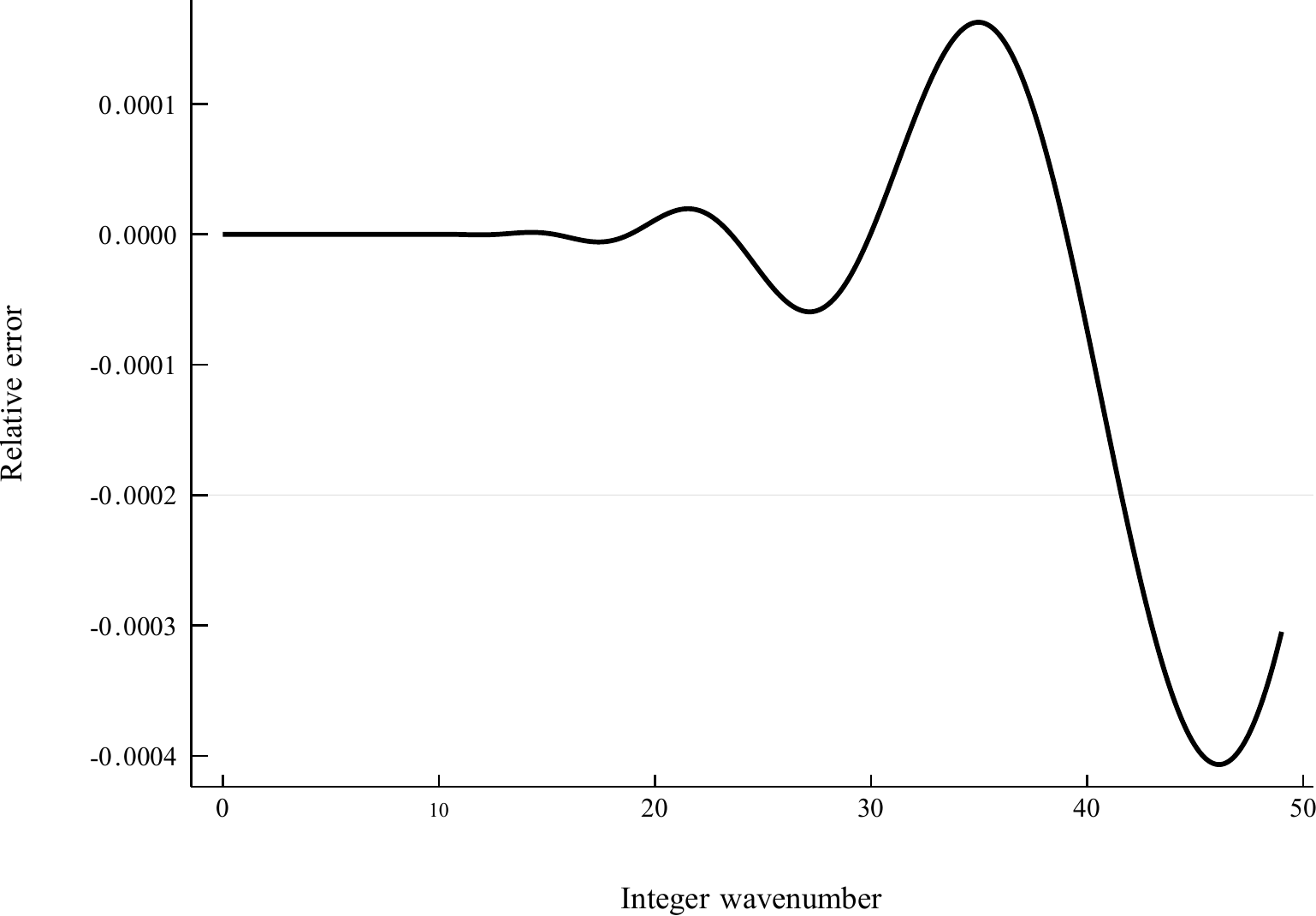}
    \caption{Relative error of a Gaussian approximation in the form \cref{eq:gaussian-mra} of the Fourier-space Green's function $(1-\ell^2k^2)^{-1/\beta}$ with parameters $\ell=1.0$, $\beta=0.5$, $h=0.2$, $M=32$, and $N=28$.}
    \label{fig:gaussian-approx-error}
\end{figure}

With approximations of the data and the Green's function now constructed in terms of $d$-variate isotropic normalized Gaussians, convolution of which is trivial, applying a discrete version of the integral operator that inverts $\mathscr{D}$ is just:
\begin{align} \label{eq:conv}
    \mathscr{D}^{-1} y &\approx \left( \sum_{i=-M_-}^{M_+} c_i \phi (\cdot; 0, \rho_i I)  \right)
        * \left( \sum_{j=1}^{N_y} b_j \phi (\cdot; \mu_j, \xi I) \right) \\
                &= \sum_{i=-M_-}^{M_+} \sum_{j=1}^{N_y} c_i b_j \phi (\cdot; 0, \rho_i I) * \phi (\cdot; \mu_j, \xi I) \\
                &= \sum_{j=1}^{N_y} b_j \left(\sum_{i=-M_-}^{M_+} c_i \phi (\cdot; \mu_j, (\rho_i+\xi_j) I) \right). \label{eq:rbf-blurred}
\end{align}
Evaluating this quantity at the observation locations yields the output of our blur.
The last line in the manipulation above shows that the continuous function we evaluate to arrive at outputs can be interpreted as an RBF interpolant of the blurred data in terms of a new blurred basis function $\tilde{\psi}$, given by
\begin{equation}
    \tilde{\psi}(\cdot) = \sum_{i=-M_-}^{M_+} c_i \phi \left(\cdot; 0, (\rho_i+\xi_j) I\right). 
\end{equation}
The weights $b_j$ of the blurred interpolant are identical to the weights of the input interpolant, which will be valuable later in this section.

We now describe our blurring algorithm more concretely. Algorithm 2.1 ties together pieces of the Green's function approximation specified in \cref{eq:mclean-int}-\cref{eq:green-approx-end}. Algorithm 2.2 combines RBF interpolation of the data with the output of Algorithm 2.1 with a convolution and evaluates the result at the data locations. These algorithms, used together, are a complete description of our blur.

\begin{algorithm}
\label{alg:mclean}
\begin{algorithmic}
\caption{Gaussian approximation of fractional bound-state Helmholtz kernel}
\Input
\Desc{Blurring scale parameter $\ell > 0$.}
\Desc{Blurring shape parameter $\beta > 0$.}
\Desc{Integration mesh size $h>0$.}
\Desc{Number of negative integration steps $M_- \in \mathbb{N}$.}
\Desc{Number of positive integration steps $M_+ \in \mathbb{N}$.}
\Desc{Dimension of each measurement location $d \in \mathbb{N}$.}
\EndInput

\Output
\Desc{Vector of positive weights $\bm{c} \in \mathbb{R}^{M_-+M_++1}$.}
\Desc{Vector of variances $\bm{\rho} \in \mathbb{R}^{M_-+M_++1}$.}
\EndOutput

\Function{GaussianBSH}{$\ell$, $\beta$, $h$, $M_-$, $M_+$, $d$}

\For{$n = -M_- \to M_+$}
\State $\hat{a} \gets \exp \left( nh - \exp \left( -nh \right) \right)$.
\State $\hat{w} \gets h \left( 1 + \exp \left( -nh \right) \right) \exp \left( \beta \left( nh - \exp \left( -nh \right) \right) \right)$.
\State $\hat{w}^\prime \gets \hat{w} \exp \left( -\hat{a} \right) \left( \pi / \ell^2 \hat{a} \right)^{d/2}/\Gamma(\beta)$.

\State $\rho_n \gets 2\ell^2\hat{a}$.
\State $c_n \gets \hat{w}^\prime \rho / 2\pi$.
\EndFor

\State \Return $\bm{c}, \bm{\rho}$.
\EndFunction
\end{algorithmic}
\end{algorithm}

\begin{algorithm}
\label{alg:blur}
\begin{algorithmic}[1]
\caption{Blur}
\Input
\Desc{Vector of measured data $\mathbf{z} \in \mathbb{R}^{N_z}$.}
\Desc{Array of data location vectors $\mathbf{q}_i \in \mathbb{R}^d$, $i\in(1, \ldots, N_z)$.}
\Desc{RBF scale parameter $\xi > 0$.}
\Desc{Vector of positive weights $\bm{\omega} \in \mathbb{R}^{M_-+M_++1}$.}
\Desc{Vector of variances $\bm{\rho} \in \mathbb{R}^{M_-+M_++1}$.}
\EndInput

\Output
\Desc{Array of blurred data $\tilde{\mathbf{z}}_i \in \mathbb{R}$, $i\in(1, \ldots, N_z)$}
\EndOutput

\Function{Blur}{$\mathbf{z}$, $\mathbf{q}$, $\xi$, $\bm{\omega}$, $\bm{\rho}$}

\State Let $\phi(\|\cdot\|; 0, \xi \mathbf{I})$ be a unit-mass Gaussian to use as the RBF kernel.

\State Generate RBF weight matrix $\mathbf{B}$ with elements:
\For{$i = 1 \to N_z, j = 1 \to N_z$}
\State $B_{ij} \gets \phi(\|\mathbf{z}_i - \mathbf{q}_j\|; 0, \xi \mathbf{I})$.
\EndFor

\State $\mathbf{b} \gets \mathbf{B}^{-1}\mathbf{z}$.

\For{$i = 1 \to N_y$}
\State $\tilde{z}_i \gets \sum_{i',n} (b_i \cdot \omega_n) \phi(\|y_i - \bm{q}_i'\|;0,(\xi + \rho_n)I)$.
\EndFor

\State \Return $\tilde{\mathbf{z}}$.

\EndFunction
\end{algorithmic}
\end{algorithm}

Algorithm 2.2 defines a linear operator $\mathbf{S}$ on $\mathbb{R}^{N_z}$.
We will prove that $\mathbf{S}$ is positive definite in \cref{thm:pos-def}, which will be useful in \cref{sec:ssir}.
The theorem is more general than the specific algorithm so far presented, which will set the stage for potential variants to be described in \cref{sec:discussion}.
To ease into the theorem, we will summarize the preceding development of $\mathbf{S}$ and connect it to a briefer alternative formulation that is easier to treat analytically.

We described a sequence of mappings between vector spaces, with blurring taking place most explicitly in the function space $L^2(\mathbb{R}^{N_y})$ of interpolants by way of the convolution $\psi_i \mapsto \mathscr{G}\psi_i$, where $\mathscr{G}$ denotes an operator that performs convolution with the Gaussian approximation of $g$, and $\psi_i$ is defined as the interpolation basis function $\psi(\|\cdot - \mathbf{q}_i\|)$ centered at location $\mathbf{q}_i$.
\begin{samepage}
Taken literally, that conceptual development prescribes the following composition of linear operations:
\begin{align}
    Y \xrightarrow{\mathbf{B}^{-1}} W \xrightarrow{\mathscr{F}} X \xrightarrow{\mathscr{G}} \tilde{X} \xrightarrow{\tilde{\mathscr{F}}^{-1}} \tilde{W} \xrightarrow{\tilde{\mathbf{B}}} \tilde{Y}.
\end{align}
Nodes in this diagram represents the various vector spaces found along the way of describing our blurring algorithm:
\begin{itemize}
    \item $Y$ is the space of input data,
    \item $W$ is the space of interpolant weights in the basis $\{ \psi_i \}$,
    \item $X = \vspan \{ \psi_i \} \subset L^2(\mathbb{R}^n)$ is the space of interpolants,
    \item $\tilde{X} = \vspan \{ \mathscr{G}\psi_i \} \subset L^2(\mathbb{R}^n)$ is the space of blurred interpolants,
    \item $\tilde{W}$ is the space of blurred interpolant weights in the basis $\{ \mathscr{G}\psi_i \}$, and
    \item $\tilde{Y}$ is the space of blurred data.
\end{itemize}
\end{samepage}
\begin{samepage}
Arrows in the diagram represent the action of the operators superscribed on them:
\begin{itemize}
    \item $\mathbf{B}^{-1}$ maps input data to RBF weights,
    \item $\mathscr{F}$ maps RBF weights to interpolated functions,
    \item $\mathscr{G}$ maps interpolated functions to blurred functions,
    \item $\tilde{\mathscr{F}}^{-1}$ maps blurred functions to weights in a blurred RBF basis, and
    \item $\tilde{\mathbf{B}}$ maps blurred weights to blurred data.
\end{itemize}
\end{samepage}

The complicated sequence of steps above can simplify greatly; observe in \cref{eq:rbf-blurred} that the weights in the blurred basis $\{ \mathscr{G}\psi_i \}$ are always identical to the weights in the unblurred basis $\{ \psi_i \}$.
Therefore $\tilde{\mathscr{F}}^{-1}\mathscr{G}\mathscr{F} = \mathbf{I}$, leaving just $\mathbf{S} = \tilde{\mathbf{B}}\mathbf{B}^{-1}$ where $\tilde{\mathbf{B}}$ is the RBF matrix in the blurred RBF basis:
\begin{equation}
    \tilde{B}_{i,j} = \tilde{\psi}(\|\mathbf{q}_i-\mathbf{q}_j\|).
\end{equation}
This alternative perspective demonstrates that $\mathbf{S}$ is equivalent to finding weights $\mathbf{b}$ for an RBF interpolant of the unblurred data using a basis $\{\psi_i\}$, and then evaluating an RBF interpolant of the \emph{blurred} data using the same weights $\bm{b}$ that now act as coefficients on a blurred basis $\{ \mathscr{G} \psi_i \}$.
The following theorem is stated in terms of this simplified perspective.

\begin{theorem} \label{thm:pos-def}
Let $\psi : \mathbb{R}^{+} \rightarrow \mathbb{R}$ be an interpolating radial basis function and let $g : \mathbb{R}^{+} \rightarrow \mathbb{R}$ be a convolution kernel.
Suppose $\psi$ and $g$ each have positive Fourier transforms, and define $\tilde{\psi} = g*\psi$.
Then the matrices $\mathbf{B}$ with entries $B_{ij} = \psi(\|\mathbf{q}_i - \mathbf{q}_j\|)$ and $\tilde{\mathbf{B}}$ with entries $\tilde{B}_{ij} = \tilde{\psi}(\|\mathbf{q}_i - \mathbf{q}_j\|)$ are symmetric positive definite, and the product $\mathbf{S}=\tilde{\mathbf{B}} \mathbf{B}^{-1}$ is positive definite.
\end{theorem}
\begin{proof}
A standard theorem of RBF interpolation (e.g.~Section 3 of \cite{FF15}) states that $\mathbf{B}$ is positive definite under the assumption that the Fourier transform of $\psi$ is positive.

The Convolution Theorem guarantees that $\tilde{\psi} = g*\psi$ has a positive Fourier transform if $g$ and $\psi$ both have positive Fourier transforms, so $\tilde{\mathbf{B}}$ is positive definite for the same reason that $\mathbf{B}$ is positive definite.

Observe that $\mathbf{B}$ and $\tilde{\mathbf{B}}$ are symmetric by construction, and that $\mathbf{B}^{-1}$ is symmetric positive definite since it is the inverse of a symmetric positive definite matrix.
Theorem 7.6.3 in \cite{Horn10} states that the product of a positive definite matrix $\mathbf{P}$ and a Hermitian matrix $\mathbf{Q}$ is a matrix with the same number of negative, zero, and positive eigenvalues as $\mathbf{Q}$.
It follows that the product of two Hermitian positive definite matrices is also positive definite.
Therefore, since $\tilde{\mathbf{B}}$ and $\mathbf{B}^{-1}$  are Hermitian positive definite matrices, $\tilde{\mathbf{B}}\mathbf{B}^{-1}$ is positive definite.
\end{proof}

The coefficients in the multiresolution approximation \cref{eq:green-approx} are all positive, and the Fourier transform of a positive Gaussian is also a positive Gaussian.
Therefore $\mathbf{S}$, as defined by Algorithm 2.2 together with Algorithm 2.1, is positive definite as a corollary of \cref{thm:pos-def}.

\section{Example 1: circular measurement locations embedded in a 2-plane}
\label{sec:ex1}
We want to verify that the blur's effect resembles what we would expect of a discrete approximation to $\mathscr{D}^{-1}$.
To that end, we blur equally-spaced data on a circle embedded in $\mathbb{R}^2$ to provide insight into the spectral properties of the blur in practice.
This example will also describe heuristics in choosing parameters $\ell$, $\beta$, and $\xi$. 
In the course of this example we will also demonstrate an undesirable phenomenon whereby $\mathbf{S}$ attenuates even the largest scales, and suggest a workaround.

Locations were chosen to encircle the origin in $\mathbb{R}^2$ with $N_y=100$ distinct locations separated by unit distance from nearest neighbors, i.e.
\begin{align*}
\mathbf{q}_i &= \left| e^{2\text{i}n\pi/100}-1 \right|^{-1} \begin{bmatrix}
\cos \left( 2n\pi/100 \right) \\
\sin \left( 2n\pi/100 \right)
\end{bmatrix}.
\end{align*}
The interpolation kernel was chosen to be the isotropic Gaussian PDF with standard deviation $\xi^{1/2}=2.5$.
The convolution kernel is the multiresolution Gaussian approximation \cref{eq:green-approx} to the fractional bound-state Helmholtz kernel with $\ell=1$ and $\beta=1$, using approximation parameters $h=0.2$, $M=32$, and $N=28$.
These parameters yield an approximation of $\hat{g}$ with $<0.05\%$ relative error up to $k_{max}=49$, the Nyquist number for the one-dimensional problem that this example simulates embedded in two dimensions.
Recall that \cref{fig:gaussian-approx-error} shows this relative error as a function of $k$.

The blur thus constructed defines a linear operator $\mathbf{S}$ on $\mathbb{R}^{N_y}$.
For the purpose of inspecting the effect of blurring at different scales we numerically constructed a matrix representation of $\bm{S}$ (though for practical application of the blurring algorithm it is inadvisable to actually construct $\mathbf{S}$).
Due to the rotational symmetry of observation locations, the matrices $\bm{B}^{-1}$ and $\tilde{\bm{B}}$ are circulant.
The class of circulant matrices is stable under inversion, transposition, and matrix multiplication, so $\bm{S}= \tilde{\bm{B}} \bm{B}^{-1}$ is also circulant.
Therefore eigenvectors of $\bm{S}$ are discrete Fourier vectors.
This connection provides a rationale for comparing eigenvalues of $\bm{S}$ to the spectrum of $\mathscr{D}^{-1}$, whose eigenfunctions are Fourier modes.
However, it is important to recognize that the comparison is imprecise because the interpolant \cref{eq:rbf} represents these discrete Fourier eigenvectors as functions that differ from Fourier modes on $\mathbb{R}^2$.

Eigenvalues of $\bm{S}$ are plotted in \cref{fig:blurring_operator_eigvals} as circles, and some examples of eigenfunctions of $\bm{S}$ are visualized beneath the plot for $k\in(1,2,25,49)$.
\emph{Eigenfunctions of} $\bm{S}$ are defined here as continuous interpolants of the matrix's eigenvectors found by the RBF interpolation scheme utilized in the blur.
This figure also shows a solid trace labelled ``Fourier'' that plots $(1+\ell^2k^2)^{-\beta}$, which is the spectrum of $\mathscr{D}^{-1}$ that corresponds to $\mathbb{R}^2$ Fourier modes.
Since eigenvectors of $\bm{S}$ do not correspond to $\mathbb{R}^2$ Fourier modes, this trace of the Fourier spectrum is only a rough comparison rather than an analytical prediction that we are trying to match.
The observed spectrum of $\mathbf{S}$ behaves as expected, with gradual blurring of small scale features.

\begin{figure}
    \centering
    \includegraphics[width=\columnwidth]{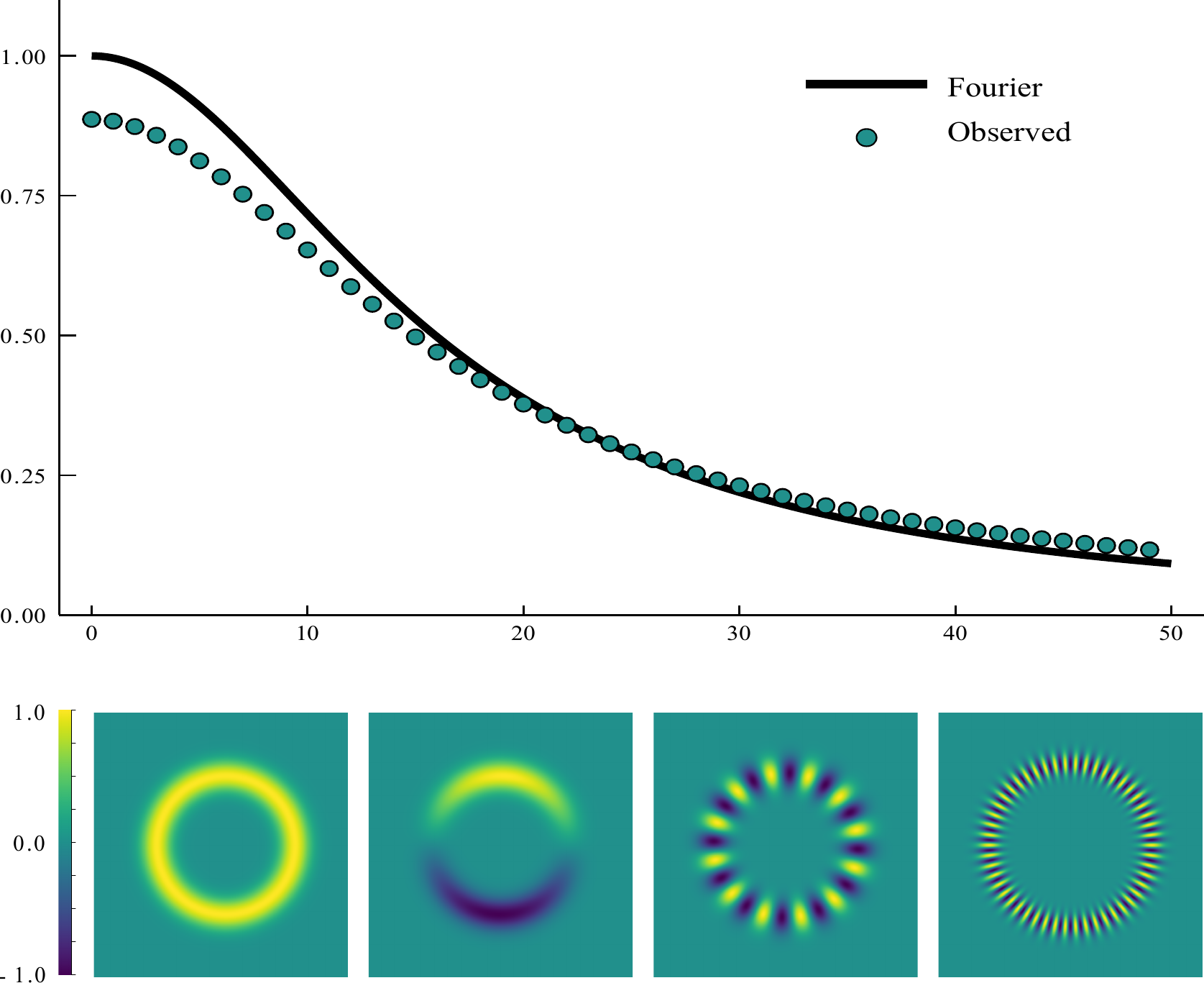}
    \caption{Top: points indicate eigenvalues of the matrix $\mathbf{S}$ defined using interpolation by Gaussian radial basis functions of standard deviation $\xi^{1/2}=2.5$,
    followed by convolution with a Gaussian approximation of the Green's function for the bound-state fractional Helmholtz kernel of $\mathscr{D}=(1-\Delta)^\beta$ with $\ell=1.0$ and $\beta=1$, acting on 100 equally spaced points around the origin in $\mathbb{R}^2$ with unit nearest-neighbor distance.
    The solid trace shows the spectrum $(1+k^2)^{-2\beta}$ of $\mathscr{D}^{-1}$, eigenfunctions of which are Fourier modes; this serves to highlight the similarity between the spectra of $\bm{S}$ and of $\mathscr{D}^{-1}$, but is not an analytical solution to match since interpolating the eigenfunctions of $\bm{S}$ are not Fourier modes in the plane.
    Bottom: some example eigenfunctions, defined as interpolants given by \cref{eq:rbf} of the eigenvectors of $\bm{S}$, for $k \in \{ 1, 2, 25, 49 \}$. Duplicate eigenpairs that arise due to symmetry are suppressed in this figure.}
    \label{fig:blurring_operator_eigvals}
\end{figure}

Recall from \cref{sec:method} that this method has a drawback of attenuating large scales.
That behavior is evident in the eigenvalues plotted in \cref{fig:blurring_operator_eigvals}, which are all less than 1.
Eigenvalues less than 1 correspond to attenuation, so the blur attenuates even the largest scale (eigenmode index 0).
This over-attenuation occurs because convolution with $g$ attenuates every Fourier eigenmode of $\mathscr{D}$ except constant functions in $\mathbb{R}^{d}$.
Since a finite Gaussian approximation can never fully describe a nonzero spatial constant, even the largest-scale function in the space of possible RBF interpolants will be attenuated by our blur.
The largest-scale eigenfunctions in this example are thin in the direction transverse to the circle, causing those modes to be blurred more than continuous Fourier modes in $\mathbb{R}^2$ with the same wavenumber (i.e.~Fourier modes with \emph{planar} length scale equal to the \emph{circumferential} length scale of $\bm{S}$ eigenfunctions)

For a similar reason that large scales are attenuated too much, the blur does not suppress the smallest-scale eigenmodes as much as $\mathscr{G}$ would suppress a true $\mathbb{R}^2$ Fourier mode of the same wavenumber.
This is because the RBF interpolants of the most highly-oscillatory eigenvectors in this example have more large-scale content than $\mathbb{R}^2$ Fourier modes with the same length scale.

Over-attenuation can be mitigated, so that the largest scales are closer to unity, by rescaling the operator by replacing $\mathbf{S} \mapsto \mathbf{S} / \|\mathbf{S1}\|$, where $\mathbf{1}$ is a unit-norm vector with all entries identical.
The eigenvectors of $\mathbf{S}$ are usually \emph{not} discrete Fourier vectors like they are in this symmetric example, so the largest-scale eigenvector is not necessarily $\mathbf{1}$.
Therefore this mitigation technique is only a heuristic, which derives from the idea that an input with identical entries contains little small-scale information.

Choosing the RBF standard deviation parameter $\xi^{1/2}$ is not to be taken lightly.
We recommend choosing it to be roughly on the order of the nearest-neighbor distance between measurements.
A value too small prevents the RBF interpolation step from resolving gradual transitions from location to location, causing the interpolant to appear as a rugged set of ``spikes'' that are overly suppressed by the convolution step on account of their inappropriately small scale.
Choosing an interpolation kernel that is too large, however, can cause numerical problems related to ill-conditioning of the linear system we must solve to arrive at RBF coefficients.
Choosing $\xi^{1/2}$ to be as large as possible, balanced against insurmountable instability due to ill-conditioning, is considered a best practice in RBF literature \cite{FF15}.



\section{Example 2: scale separation of ocean temperature measurements}
\label{sec:scale-separation}
The application of our algorithm is demonstrated with data from Argo, an international instrumentation project to measure the world's oceans \cite{ARGO}.
Argo uses an array of approximately 3600 profiling floats.
Every 10 days, these devices sink to a depth of 2000 meters to capture a depth-wise profile of measurements as they float back to the surface.

\begin{figure}
    \centering
    \includegraphics[width=0.85\columnwidth]{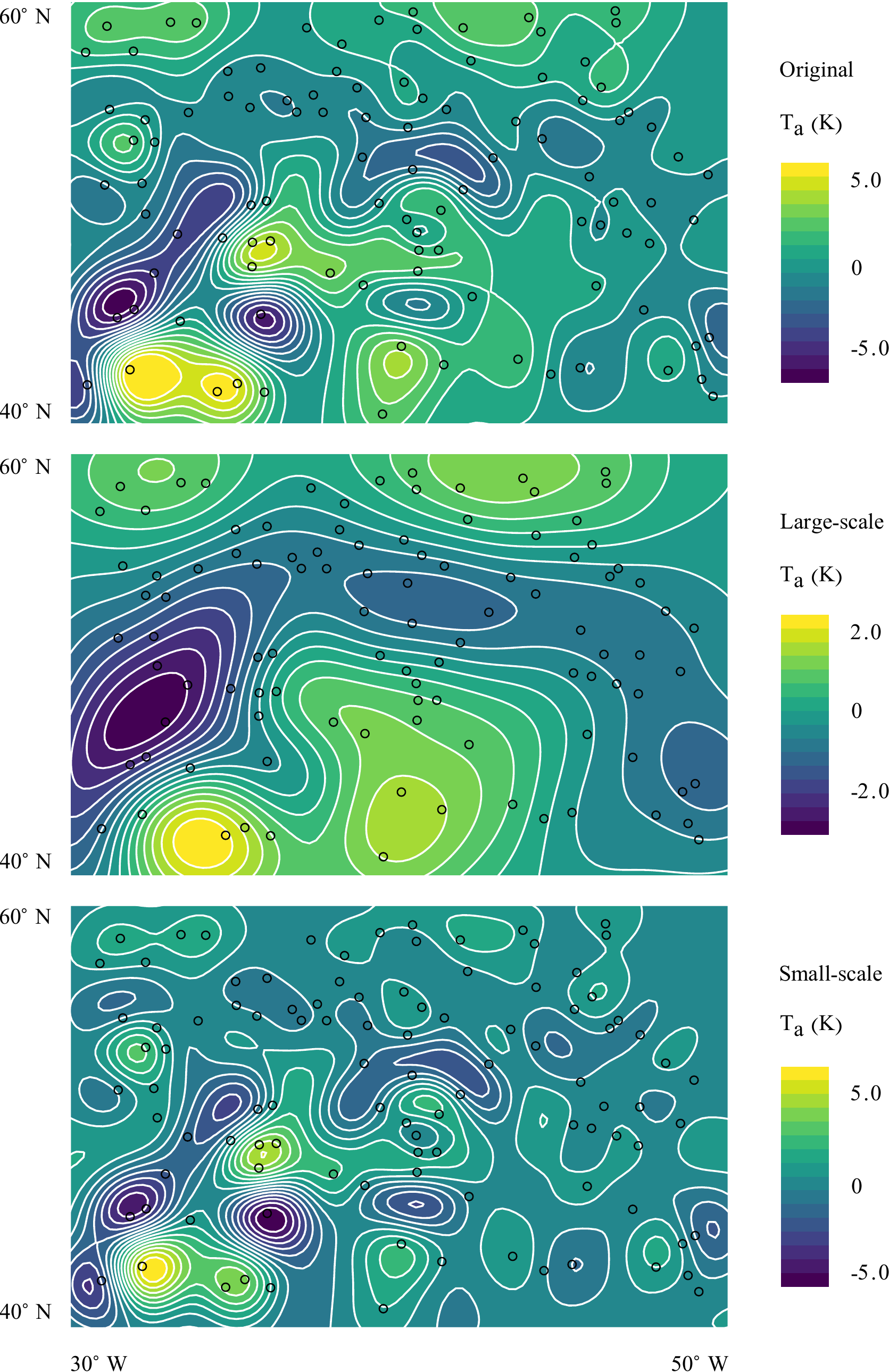}
    \caption{
    Sea surface temperature measurements made by floats in the Argo oceanographic instrumentation project.
    The horizontal and vertical coordinates describe longitudinal and latitudinal displacement, in kilometers, from a center of $50^\circ$N $35^\circ$W.
    Contours and colors indicate interpolated temperature differences in degrees Celsius, as estimated by the original interpolant of the data used in the blurring algorithm with $\xi^{1/2}=175km$, $\ell=70km$, and $\beta = 8$ (top); by the blurred interpolant representing large-scale features (center), and the small-scale interpolant defined as the difference of the blurred interpolant subtracted from the original interpolant (bottom).
    Although the linear fit is an important part of both the original and large scales, it is not added back for the purpose of this visualization in order to show other large-scale features that would otherwise be difficult to see.
    }
    \label{fig:argo}
\end{figure}

We apply our blurring algorithm to Argo sea surface temperature measurements to decompose the data into large-scale and small-scale components.
Physical oceanographers are interested in how ocean heat content varies across spatial and temporal scales, and the Argo data set is an invaluable tool \cite{wortham2014multidimensional}.
Given the irregularity of the Argo data in space and time, it is usually interpolated to a regular grid before use \cite{roemmich20092004}.
Our blurring algorithm enables extraction of the large-scale (blurred) part of the data without the need to first interpolate onto a grid.
The data used in this example, accessed on February 13, 2020, were taken from the region $40^\circ$N--$60^\circ$N and $20^\circ$--$50^\circ$W over the time period February 1, 2020 through February 10, 2020.

A small number of data points were removed to maintain a minimum nearest-neighbor separation of 50 kilometers.
A least-squares linear fit is subtracted from the remaining data to arrive at deviations.
These deviations are blurred with \cref{alg:blur} using RBF length scale $\xi^{1/2}=175$ kilometers, $\ell=70$ kilometers, and exponent $\beta=8$.
Subtracting the linear fit to obtain deviations avoids rapid variation of the interpolant that could otherwise happen near the boundaries due to the basis functions' rapid decay.

The original field of sea surface temperature deviations, as interpolated by RBF stage of our algorithm, is depicted in the top panel of \cref{fig:argo}.
The center panel shows the large-scale component, which is the result of blurring the temperature deviations.
The small scale component, plotted in the bottom panel, is the difference of the unblurred and blurred temperature deviations.

\section{Application to particle filtering}
\label{sec:ssir}
The blur described above was originally motivated by an effort to mitigate the dimensional curse of particle filtering, which we will refer to as sequential importance sampling with resampling (SIR) hereinafter.\footnote{Sequential importance sampling (SIS) is also known as \emph{particle filtering}.
In this paper we specialize on SIS with resampling (SIR), the most famous variety of particle filter, though the dimensional curse and presumably our applications also extend to other particle filter varieties.}
This section will describe how the blur was derived from the motivating considerations about SIR, in terms of an observation error covariance matrix.

SIR begins by running an ensemble of forecasts; each ensemble member is a \emph{particle}. 
A set of observations of the true system is then taken; the observations are corrupted by instrument errors, whose distribution defines a likelihood.
This likelihood is used to compute weights for each particle so that the weighted ensemble approximates the Bayesian posterior.

SIR is susceptible to a phenomenon called \emph{collapse}, characterized by essentially all the ensemble weight accumulating on a single ensemble member that is closest to the observations, causing the filter to catastrophically underestimate posterior dispersion.
The number of ensemble members required to avoid SIR collapse depends on system covariance and observation error covariance, scaling exponentially in an effective system dimension.

Specific estimates of ensemble size required to avoid collapse, provided in \cite{SBBA08,SBM15}, suggest one can reduce the required ensemble size by increasing eigenvalues of the observation error covariance.\footnote{Decreasing eigenvalues of the system covariance has the same effect, for the same reasons. However, it is harder to justify changes to a dynamical model, and to implement those changes, than to modify the observation model as we propose.}
Doing so carefully can also improve uncertainty quantification for a fixed number of ensemble members.
Ref. \cite{RGK18} suggests inflating the observation error variance at small scales, letting variance grow in wavenumber, since small scales have very limited predictability in geophysical flows \cite{Lorenz69,RS08,Judt18}.

To be more precise, consider an observing system
\begin{align}
    y(\bm{q}) &= \mathscr{H}\{x\}(\bm{q}) + r^{1/2}(\bm{q})~\epsilon(\bm{q}), \label{eq:cont-obs-sys}
\end{align}
where $y(\bm{q}) \in \mathbb{R}$ is the observation at location $\bm{q} \in \mathbb{R}^{d}$,
$\mathscr{H}$ is a function-valued observation operator acting on $x$, which describes the scalar system state as a function of location, $r^{1/2}(\bm{q})$ can be imagined as the standard deviation of the observation error at $\bm{q}$, and $r^{1/2}(\bm{q})~\epsilon(\bm{q})$ is the random observation error.

It is natural to think of $\epsilon$ as a random field.
But letting the spectrum grow in wavenumber precludes pointwise definition of $\epsilon$, with probability 1, so it is not a random field in the traditional sense.
The idea of imposing a correlation structure with a growing spectrum can instead be understood in the framework of generalized random fields.
In this case $\epsilon$ can be treated as a random process with realizations taking the form of tempered distributions, i.e.~elements of the topological dual to a Schwartz space of rapidly decaying functions on $\mathbb{R}^d$.

Since realizations of $\epsilon$ with a growing spectrum cannot be described pointwise, instead interpret the spatially parametrized terms in \cref{eq:cont-obs-sys} as averages with respect to a Schwartz function $\nu$ that is closely concentrated near $\bm{q}$.
For example,
\begin{align}
    \epsilon(\bm{q}) \equiv \left. \int_{\mathbb{R}^d} \epsilon~\nu~\text{d}x \middle/ \int_{\mathbb{R}^d} \nu~\text{d}x \right. .
\end{align}
Narrowing our attention within the scope of generalized random fields, let $\epsilon$ be a mean-zero stationary Gaussian generalized random field (GGRF).
Then the vector $\left(\epsilon(\bm{q}_1), \cdots , \epsilon(\bm{q}_{N_y}) \right)$ is a multivariate normal random variable with zero mean and a covariance matrix $\bm{C}$ with entries $C_{ij}$ that depend only on $\|\bm{q}_i-\bm{q}_j\|$.
Hence the vector of observations $\bm{y} \equiv \left( y(\bm{q}_1), \cdots , y(\bm{q}_{N_y})\right)$ conditioned on $x$ is a multivariate normal random variable with mean $H(\bm{x})$ and covariance
\begin{align*}
	\mathbf{R} = \mathbf{R}_0^{1/2}\mathbf{C}\mathbf{R}_0^{1/2},
\end{align*}
where $\bm{R_0}^{1/2}$ is a diagonal matrix of the discrete observation standard deviations $r^{1/2}(\bm{q}_i)$ that can be treated as instrument errors and $H(\cdot) : \mathbb{R}^{N_y} \rightarrow \mathbb{R}^{N_y}$ is an observation operator acting on the discrete vector $\mathbf{x}$ that characterizes the underlying system state.
The discrete observing system can be summarized in the form
\begin{align}
    \mathbf{y} = \mathbf{H} \mathbf{x} + \bm{R}_0^{1/2} \bm{\epsilon} \in \mathbb{R}^{N_y}.
\end{align}

In the spirit of \cite{LRL11,RH05}, a connection between elliptic stochastic partial differential equations and random fields enables us to make use of fast algorithms for PDEs in the context of solving for the likelihood under GGRF models, rather than naively developing a dense approximation of $\mathbf{C}$ and then solving the associated linear system.
In seeking to build a GGRF error model with a likelihood that is cheap to evaluate in high dimensions, it was this connection that led to the development of the more generally-applicable blurring algorithm presented in this paper.

We specifically treat the continuous field of observation error as $\epsilon=\mathscr{D}\mathcal{W}$, where $\mathscr{D} = (1-\ell^2 \Delta)^{\beta}$ acts on a spatial white noise $\mathcal{W}$ with mean zero and unit pointwise variance.
As in \cref{eq:diff_op}, $\Delta$ is the formal Laplacian operator, $\ell>0$ is a tuning parameter with dimensions of length, and $\beta>0$ is a dimensionless tuning parameter that controls the rate of growth of eigenvalues.
Recall that eigenfunctions of $\mathscr{D}$ are Fourier modes of wavenumber $k$ and corresponding eigenvalues $(1+\ell^2|k|^2)^{\beta}$.
Recall also that the characteristic scale of this operator is $\ell/(2\pi \sqrt{2^{1/\beta}-1})$, in the sense that eigenfunctions with length scales longer than this have corresponding eigenvalues close to 1.
Modeling the observation error $\epsilon$ in this manner therefore ascribes a variance to large scales that is commensurate with instrument error, but it also progressively and unboundedly inflates variance for small scales at a rate controlled by $\beta$.
The GGRF description of observation error is thus a kind of surrogate model for the assumption of uncorrelated observations at large scales, but with inflated variance at small scales that are of lesser concern in geophysical forecasting.

We will use the fact that preferentially inflating observation variance at small scales is equivalent to treating blurred innovations as uncorrelated.\footnote{The \emph{innovation} of an ensemble member is the difference between observation and forecast.}
To see this equivalence, observe how the correlation matrix features in the Gaussian likelihood, the logarithm of which is proportional to
\begin{align}
    (\mathbf{y} - \mathbf{H}\mathbf{x})^T \mathbf{R}_0^{-1/2} \mathbf{C}^{-1} \mathbf{R}_0^{-1/2} (\mathbf{y} - \mathbf{H}\mathbf{x}).
\end{align}
Then consider preprocessing the ``standardized innovations'' $\mathbf{R}_0^{-1/2} (\mathbf{y} - \mathbf{H}\mathbf{x})$ with a linear operation 
\begin{align}
    \mathbf{R}_0^{-1/2} (\mathbf{y} - \mathbf{H}\mathbf{x}) ~~ \longmapsto ~~ \mathbf{S} \mathbf{R}_0^{-1/2} (\mathbf{y} - \mathbf{H}\mathbf{x}).
\end{align}
If these blurred observations are now assimilated under the assumption that the errors in the blurred field are standard normal, then the log-likelihood is proportional to
\begin{align}
    (\mathbf{y} - \mathbf{H}\mathbf{x})^T \mathbf{R}_0^{-1/2} \mathbf{S}^T\mathbf{S} \mathbf{R}_0^{-1/2} (\mathbf{y} - \mathbf{H}\mathbf{x}).
\end{align}

Any valid covariance matrix must be symmetric and positive definite.
If $\mathbf{S}$ is a positive definite blurring operator --- i.e.~a positive definite operator with a decaying spectrum toward small scales --- then $\mathbf{C} = (\mathbf{S}^T \mathbf{S})^{-1}$ is a symmetric positive definite operator with a spectrum that grows toward small scales.
In light of these considerations, take $\mathbf{S}$ to be the algorithm presented in \cref{sec:method}.
Although \cref{thm:pos-def} shows that $\mathbf{S}$ is positive definite, it is typically \emph{not} symmetric, and the symmetric part of $\mathbf{S}$ is not necessarily positive definite.
For this reason we must treat $\mathbf{R}_0^{1/2}(\mathbf{S}^T\mathbf{S})^{-1}\mathbf{R}_0^{1/2}$ as a covariance matrix, rather than $\mathbf{R}_0^{1/2}\mathbf{S}^{-1}\mathbf{R}_0^{1/2}$.

In summary, we propose blurring standardized innovations with $\mathbf{S}$ in such a way that $(\mathbf{S}^T\mathbf{S})^{-1}$ is a covariance for a discretized GGRF.

\section{Example 3: SIR with radiosonde data}
\label{sec:ex-radiosonde}

To demonstrate the behavior of our blurring algorithm on scattered data and its impact on SIR weights, we make use of data from the U.S.~National Center for Atmospheric Research (NCAR) Convection Allowing Ensemble \cite{SRSFW15,SRSFW19}.
The NCAR ensemble produced real-time 48 hour forecasts over the conterminous United States (CONUS) from April 7, 2015 to December 30, 2017.
The ensemble forecasting system consisted of two components: an 80 member ensemble assimilation system operating at 15 km resolution and a 10 member ensemble forecast system operating at 3 km resolution.
We make use of the 80 member ensemble data.
The assimilation system used the Advanced Research version of the Weather Research and Forecasting (WRF) model; observations were assimilated in a 6 hour cycle via the Ensemble Adjustment Kalman Filter \cite{Anderson03} implemented in the Data Assimilation Research Testbed software suite \cite{DART}.
Every assimilation cycle processed between 66,000 and 70,000 observations from a variety of sources including radiosondes, aircraft measurements, satellite wind measurements, and Global Positioning System radio occultation data, among others.
Further details are provided in \cite{SRSFW15}.

To verify that our blurring algorithm performs as expected on scattered data, we apply it to radiosonde temperature measurements at a single pressure level.
Every 12 hours, i.e.~every other assimilation window, there are between 90 and 97 radiosonde measurements scattered across North and Central America and the Caribbean available at various pressure levels.
An example of the locations of these observations at a pressure level of 70 kPa on May 15, 2017 is shown in Fig.~\ref{fig:radiosonde}.
The left panel shows the locations of the measurements along with an interpolated temperature field obtained using Gaussian RBFs with standard deviation $\xi^{1/2}=5^\circ$.
The right panel shows the result of applying our blurring algorithm with blurring exponent $\beta=1/2$ and blurring length scale $\ell=4^\circ$.
Figure \ref{fig:radiosonde} provides visual evidence that our algorithm indeed blurs scattered data.
Interpolating the raw data would leave strange regions of approximately zero Kelvins in the interpolants depicted.
So for the purpose of visualization, we subtract the mean before applying the blur, and then add the mean back to the blurred data.

\begin{figure}
    \centering
    \includegraphics[width=\columnwidth]{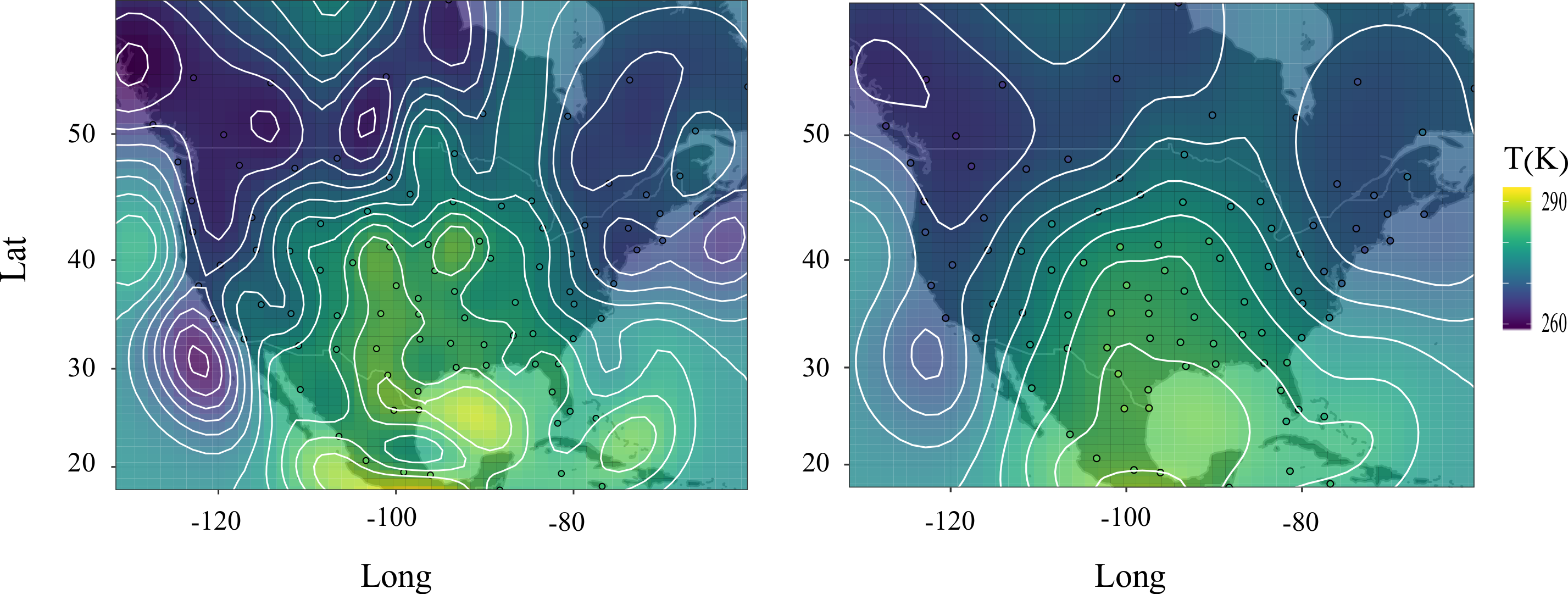}
    \caption{Left: Temperature fields interpolated from radiosonde data measured on May 15, 2017 at a pressure level of 70kPa. Right: the result of applying our blur with parameters $\xi^{1/2}=5^{\circ}$, $\beta=1$, and $\ell=4^\circ$. The shape of the North America is underlaid to give a sense of scale, and small circles indicate measurement locations.}
    \label{fig:radiosonde}
\end{figure}

We next verify that the algorithm has a controllable degree of blurring with the desired effect on SIR weights, viz. that the effective sample size increases as the blurring length scale $\ell$ increases. 
To that end we use the 80 member ensemble forecast for temperature at the locations of the radiosonde temperature observations, assume that the forecast weights are all equal to $1/80$, and update the weights based on mismatch to the observations using the standard SIR update formula.

To be precise, let $\mathbf{y}$ be the vector of radiosonde temperature observations at a given time, let $\mathbf{Hx}^{(i)}$ be the vector of forecast temperatures at the same time and locations for ensemble member $i$, and let $\mathbf{R}_0^{1/2}$ be a diagonal matrix whose diagonal contains the standard deviations of the observation errors.
The un-normalized weight for the $i^\text{th}$ ensemble member is
\begin{equation}
    \overline{w}_i = \text{exp}\left\{-\frac{1}{2\sigma}\left(\mathbf{y}-\mathbf{H}\mathbf{x}^{(i)}\right)^T\mathbf{R}_0^{-1/2}\mathbf{S}^T\mathbf{SR}_0^{-1/2}\left(\mathbf{y}-\mathbf{H}\mathbf{x}^{(i)}\right)\right\}
\end{equation}
where $\mathbf{S}$ is the matrix corresponding to the blurring operator and $\sigma = \|\bm{S1}\|^2$ is a rescaling factor with $\bm{1}$ a unit vector of identical entries.
The normalized weights are
\begin{equation}
    w_i = \frac{\tilde{w}_i}{\sum_{j=1}^N\tilde{w}_j}
\end{equation}
and the effective sample size (ESS) is
\begin{equation}
    \text{ESS }=\frac{1}{\sum_{i=1}^Nw_i^2}.
\end{equation}

\begin{figure}
    \centering
    \includegraphics[width=\columnwidth]{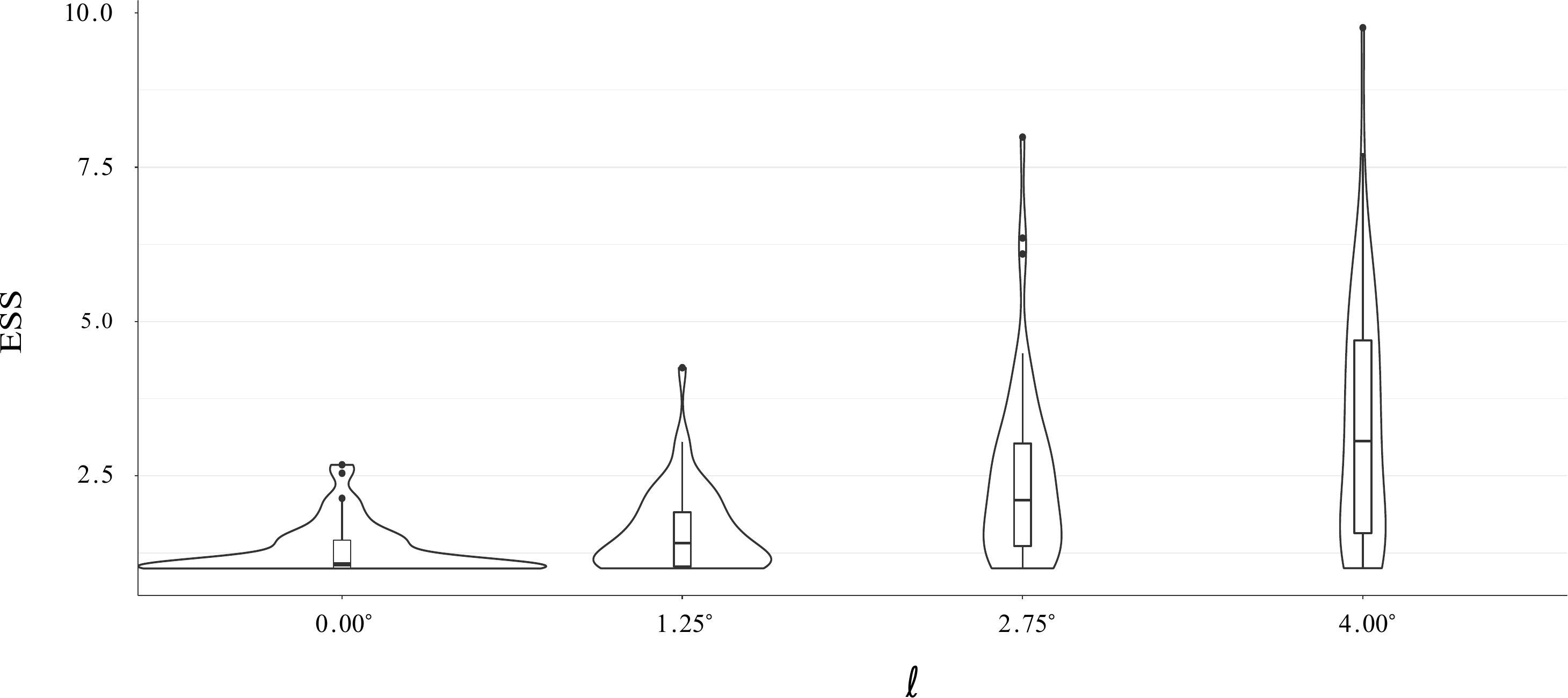}
    \caption{Distributions of effective sample sizes, for different values of blurring length scale $\ell$, observed for the posterior ensembles on radiosonde temperature.
    The posterior ensembles are obtained by performing an importance sampling update on the WRF forecast as a prior ensemble, using actual radiosonde temperature measurements at an atmospheric pressure level of 50~kPa.
    The distributions in this plot depict ESS values for SIR weights computed twice daily over the entire month of May 2017.
    }
    \label{fig:ess}
\end{figure}

Figure \ref{fig:ess} shows violin plots of the ESS computed twice daily over the entire month of May 2017 using radiosonde temperature data at a pressure level of 50 kPa.
The standard particle filter (i.e. without blurring, or $\ell=0$) exhibits very poor performance with ESS only rarely rising much beyond 1.

As the blurring length scale $\ell$ increases from 0, the distribution of ESS also increases. With $\ell=4^\circ$, the median ESS is approximately 3 and ESS occasionally rises beyond 5.
We emphasize that these values of ESS are still quite small for an 80 member ensemble, but the goal here has not been to demonstrate the performance of the particle filter per se, but of the blurring algorithm for scattered data.
That said, even a tiny increase of the ESS beyond its minimum possible value of 1 is promising because it offers hope that uncertainty quantification will improve faster upon increasing the ensemble size.

\section{Computational considerations}
\label{sec:discussion}
The first step of our blur involves solving a dense linear system \cref{eq:rbf} for interpolation weights $\bm{b}$.
A naive approach to doing so would require $\mathcal{O}(N_y^2)$ storage and $\mathcal{O}(N_y^3)$ operations.
This would be highly undesirable, especially if the interpolation matrix $\bm{B}$ changes between assimilation cycles due to changing observation locations.
Happily, there exist algorithms to solve a Gaussian RBF problem much faster.
These notably include PetRBF, which is based on GMRES iteration with a restricted additive Schwarz method preconditioner.
PetRBF requires $\mathcal{O}(N_y)$ storage and $\mathcal{O}(N_y)$ operations in arbitrary dimension $d$, and is scalable to many cores as implemented in PETSc \cite{YBK10}.

Another potential bottleneck is evaluating the sum of $N_yM$ Gaussians at $N_y$ target locations, with $M=M_-+M_++1$ terms in the approximation of $g$.
A direct approach to evaluating this sum, as is written in Algorithm 2.2, has time complexity $\mathcal{O}(N_y^2M)$.
This Gaussian sum approximation can be reduced to $\mathcal{O}(N_y+N_yM)$ with the Fast Gauss Transform (FGT) \cite{GS91}.
The FGT has exponential time complexity in location dimensionality $d$, so it often runs slower than direct evaluation when $d$ is greater than 2.
The Improved Fast Gauss Transform (IFGT) exists to eliminate that exponential scaling that hinders the original FGT for large $d$ \cite{YGDG03}.
The IFGT can be challenging to use in practice, but there exist approaches to assist in automatic tuning such as \cite{MS09}.
Finally, ASKIT offers a new approach to the kernel summation problem that extends to non-Gaussian kernels and which may be less fragile in high dimensions \cite{march2015askit}.

Our algorithm uses an interpolation basis function whose width remains fixed throughout the domain.
This may be problematic when the density of observation locations is highly heterogeneous, since an RBF standard deviation $\xi$ large enough to resolve smooth features in a sparsely-sampled region may be large enough that it causes numerical problems in densely-sampled regions.
Those numerical problems may arise from ill-conditioning of $\bm{B}$ or from insufficient data locality expected of some divide-and-conquer solvers like PetRBF.
There is some extant literature on the use of nonuniform RBF width parameters to address this situation \cite{FZ07}, so that the size of the basis function can adapt to the density of observations.
Using adaptive width involves interpolation with a basis $\{\psi_i\}$ that is allowed to vary with $i$.
Adaptive width can be incorporated into our blur, just by modifying \cref{eq:rbf} and \cref{eq:conv} to let $\xi$ vary with $i$.
Using nonuniform width parameters no longer comes with guaranteed nonsingularity, but \cite{FZ07} suggests that singularity is more of an exception than a rule.

Unfortunately, many fast solvers for the RBF problem are incompatible with basis functions that vary by location.
One possibility to reduce the cost of solving for interpolation weights in this case is to choose compactly-supported basis functions $\psi_i$ so that $\bm{B}$ is sparse.
We are unaware of any compactly-supported radial basis functions with positive Fourier transforms that are simple to convolve with a Gaussian, particularly for arbitrary $d$.
But performing the interpolation in terms of compactly-supported bases $\psi_i$ can be made compatible with the rest of our blurring method, simply by approximating each $\psi_i$ with a sum of Gaussians.
The resulting Gaussian approximation of the data will not be an interpolant, but careful construction can make it accurate.
Therefore \cref{thm:pos-def} does not apply, but we can still expect this substitution to yield a good approximation of the convolution acting on the original interpolant.

It is similarly possible to choose a different convolution kernel $g$ to approximate with a sum of Gaussians.
This idea can be used to implement a blur of the form presented here with a wider variety of characteristics, such as a non-monotonic response in length scale.
If the Guassian approximation kernel possesses a positive Fourier transform, and the RBF interpolation employs a uniform basis function, then \cref{thm:pos-def} still applies to guarantee that $\mathbf{S}^T\mathbf{S}$ is a valid covariance matrix in its application shown in \cref{sec:ssir}.

To reduce the $M$ prefactor in the convolution step, we can apply a reduction algorithm based on Prony's method with the suboptimal approximation \cref{eq:mclean-disc} as a starting point \cite{BM10}.
Doing so yields an optimal multiresolution approximation of the integral kernel for given uniform relative error bounds, which may require substantially fewer terms to attain the same relative accuracy.
This reduction method may be particularly helpful for different forms of $\mathscr{D}$ (ergo $g$) that do not yield such a rapidly-convergent approximation as \cref{eq:green-approx}.


\section{Conclusions}
\label{sec:conclusions}
We have a described a method to blur data measured at $N_y$ locations that are arbitrarily scattered in $\mathbb{R}^d$, for arbitrary $d$, by applying a discrete approximation to the integral equation that inverts the fractional bound state Helmholtz operator $(1-\ell^2\Delta)^\beta$.
The degree of attenuation for different length scales can be tuned by adjusting the parameters $\ell>0$ and $\beta>0$; large scales are attenuated little, but length scales shorter than $\ell/(2\pi \sqrt{2^{1/\beta}-1})$ are rapidly suppressed with a rate determined by $\beta$.

The discrete approximation results from a multiresolution Gaussian approximation to the differential operator's Green's function.
This readily permits convolution with a sum of Gaussians that approximate the data; we take the sum of Gaussians approximation to be a radial basis function (RBF) interpolant with a Gaussian kernel.
The blur is shown to be a positive definite linear operator on $\mathbb{R}^{N_y}$ in a more general context where the interpolation basis and the convolution kernel have positive Fourier transforms.

Spectral properties of our blur are examined with an example shown in \cref{sec:ex1}.
This example provides evidence that the algorithm operates as expected: attenuation gradually increases in wavenumber, roughly approximating the differential operator's inverse spectrum, with a caveat that our blur attenuates even the largest scale.
In order to preserve large scales, we propose dividing $\bm{S}$ by $\|\bm{S1}\|$, where $\bm{1}$ is a unit vector with all entries identical.

\Cref{sec:scale-separation} shows an example application of our blurring algorithm on sea surface temperature data measured by oceanographic floats in the Argo project.
In contrast to the typical approach of decomposing these data into large-scale and small-scale components, our method circumvents the need for interpolating onto a regular grid.

Our blur is developed with application to Sequential Importance Sampling with Resampling (SIR) particle filters in mind.
\Cref{sec:ssir} describes how blurring observations with $\mathbf{S}$ before assimilating them as if they have uncorrelated errors is equivalent to assuming that the observation errors have covariance $(\mathbf{S}^T\mathbf{S})^{-1}$, which gives observation errors the correlation structure of a stationary generalized Gaussian random field.
Relative to an uncorrelated model, an observation error model of this type decreases the number of ensemble members required to achieve good uncertainty quantification from SIR for spatially-extended dynamical systems \cite{RGK18}.

\Cref{sec:ex-radiosonde} demonstrates that this blur has the desired effect of helping balance SIR weights in an example with real meteorological data, which improves uncertainty quantification by reducing the tendency of SIR to produce underdispersed posterior distributions in high dimensions.
This example is chosen to be provocative of potential future applications to geophysical fluid dynamics, but it is worth characterizing traits of applications that would be more appropriate.
The extratropical temperature field in \cref{sec:ex-radiosonde} probably features little dynamical nonlinearity at large scales, and its measurements are linear and Gaussian, so this corpus of data is an excellent candidate for assimilation with any one of the many variants of the Ensemble Kalman Filter.
A more appropriate application of blurred-observation SIR would feature substantially non-Gaussian behavior at large scales.
That can arise due to nonlinear dynamics of large scales or due to large dispersion of a non-negative state variable relative to its mean, or due to a nonlinear observation operator inducing a non-Gaussian posterior distribution.
Moist convective systems, for example, have nonlinear dynamics and substantially skewed sign-definite variables.
Examples of nonlinear observation operators that could be similar motivation for SIR include satellite radiance and precipitation measurements.
Any of these features could provide motivation for accepting the computational challenge of SIR in exchange for provable convergence to the smoothed-observation surrogate model with its controllable bias.
However, it remains to be seen whether a smoothed-observation SIR filter can beat methods in the EnKF family when applied to real atmospheric problems.


A naive implementation of Algorithm 2.1 requires $\mathcal{O}(N_y^2)$ memory and $\mathcal{O}(N_y^3)$ operations to solve for interpolant weights.
However \cref{sec:discussion} describes how specialized kernel matrix solvers and fast kernel summation methods can reduce the asymptotic complexity of our algorithm to $\mathcal{O}(N_y)$.

\section*{Acknowledgments}
We are grateful to Jeff Anderson and Glen Romine for their help in accessing and using the NCAR ensemble data. We are grateful to Gregory Beylkin for pointing us towards a multiresolution Gaussian approximation of the blurring kernel.
The Argo data used here (\url{http://doi.org/10.17882/42182#70590}) were collected and made freely available by the International Argo Program and the national programs that contribute to it (\url{http://www.argo.ucsd.edu},  \url{http://argo.jcommops.org}).  The Argo Program is part of the Global Ocean Observing System.

\bibliographystyle{siamplain}
\bibliography{references}
\end{document}